\documentclass{article}
\usepackage[utf8]{inputenc}
\usepackage{fullpage}
\usepackage{amsfonts}
\usepackage{amsmath}
\usepackage{amssymb}
\usepackage{amsthm}
\usepackage{url}
\usepackage{booktabs}
\newcommand{\ra}[1]{\renewcommand{\arraystretch}{#1}}
\usepackage{tikz}
\usetikzlibrary{positioning, shapes.geometric}
\usepackage{graphicx}  
\newtheorem{ass}{Assumption}
\newtheorem{theo}{Theorem}

\newtheorem{mycorollary}{Corollary}

\title{Unbiased and Efficient Estimation of Causal Treatment Effects in Cross-over Trials}
\author{Jeppe Ekstrand Halkjær Madsen$^{* 12}$, Thomas Scheike$^1$, Christian Pipper$^{23}$}
\date{}

\begin{document}
\maketitle
$^1$Section of Biostatistics, Department of Public
Health, University of Copenhagen,
Copenhagen K, Denmark.

$^2$Biostatistics and Pharmacoepidemiology,
Medical Sciences, Leo Pharma A/S, Ballerup,
Denmark.

$^3$ Epidemiology, Biostatistics \& Biodemography,  Dept. of Public Health, University of Southern Denmark, Denmark.

\begin{abstract}
    We introduce causal inference reasoning to cross-over trials, with a focus on  Thorough QT (TQT) studies. For such trials, we propose different sets of assumptions and consider their impact on the modelling strategy and estimation procedure. 
    
    We show that unbiased estimates of a causal treatment effect are obtained by a G-computation approach in combination with weighted least squares predictions from a working regression model. Only a few natural requirements on the working regression and weighting matrix are needed for the result to hold. It follows that a large class of Gaussian linear mixed working models lead to unbiased estimates of a causal treatment effect, even if they do not capture the true data generating mechanism.     
    
    We compare a range of working regression models in a simulation study where data are simulated from a complex data generating mechanism with input parameters estimated on a real TQT data set. In this setting, we find that for all practical purposes working models adjusting for baseline QTc measurements have comparable performance. Specifically, this is observed  for working models that are by default too simplistic to capture the true data generating mechanism. 
    
    Cross-over trials and particularly TQT studies can be analysed efficiently using simple working regression  models without biasing the estimates for the causal parameters of interest. 
\end{abstract}


\section{Introduction}
A TQT study is an essential component of drug development, ensuring drug safety for patients. Therefore, it is a regulatory requirement to conduct such trials \cite{food_and_drug_administration_hhs_international_2005}. These trials have complicated designs in an attempt to minimize the sample size. This complexity has in turn led to a long-standing debate and several suggestions on how to best model the resulting data, and in particular how to use baseline measurements \cite{lu_efficient_2014, kenward2010use, novel,schall_mixed_2011,orihashi_concentration-qtc_2021}. In parallel with these developments, the causal inference literature, and recently also official regulatory recommendations, have increased the focus on clearly defining what we are actually trying to estimate. This has led to the endorsement of the so called estimand framework within regulatory guidance documents and the causal roadmap concept within the causal inference literature \cite{international2019addendum, van_der_laan_targeted_2011}. 

One of the central points made here is that we should enable our research question to be defined  in terms of the trial data and not just in terms of a specific model. That said, we would still want to use models in order to gain efficiency or eliminate bias. This is in line with recent regulatory guidance documents that encourage the active use of baseline variables in randomized trials for gaining precision of estimated treatment effects \cite{ema2015,fda2019}. 

The developments in this paper aim to support this push towards clearer and more focused statistical procedures. In particular, the causal inference approach we present facilitates a clear and transparent definition of our target of estimation in cross-over trials and specifically in TQT studies. 

Similar developments are already available in the literature for one-period trials. Here, standard estimators have been shown to be unbiased for causal parameters \cite{rosenblum,wang_model-robust_2021}. Moreover, appreciable gains in efficiency compared to unadjusted estimators has been demonstrated \cite{robinson_surprising_1991,hernandez_adjustment_2006,bartlett_covariate_2018}. We provide similar results for standard estimators in cross-over trials, and in particular in TQT studies.

For a standard TQT study, healthy volunteers \cite{international2019addendum} are enrolled with the purpose of obtaining electrocardiograms (ECGs) from each subject under different treatment conditions. An ECG measures the electrical output from the heart, resulting in replicates of graphical outputs as illustrated in Figure \ref{ecg}. The output is characterized by different waves, complexes, and intervals. In Figure \ref{ecg} we see the P, Q, R, S, and T waves. The interval from the beginning of the Q wave to the end of the T wave is called the QT interval, and is measured in milliseconds (ms). The QT interval measures the time it takes the heart to repolarize and prepare for the next beat. The longer the QT interval, the longer the time between heart beats, and the less oxygen is transported to cells in the body. Specifically,  prolongation of the QT interval has been shown to be related to an increased risk of Torsades de Pointes, a malignant ventricular arrhythmia. Thus, it is undesirable for the QT interval to be prolonged due to drug exposure. 
\begin{figure}[ht]
    \centering
    \includegraphics[scale=.3]{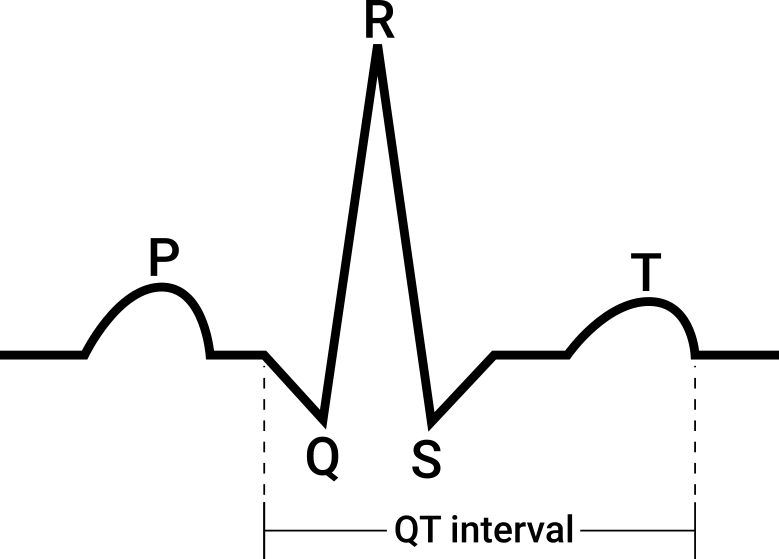}
    \caption{ECG.}
    \label{ecg}
\end{figure}
The length of the QT interval is positively associated with the length of the  RR interval, i.e. the interval from the R wave on an ECG until the next R wave on the ECG. Therefore, QT intervals are standardized in order to get a QTc (QT corrected) measurement corresponding to a particular length of the RR interval (typically 1 second). An example of a commonly used correction is the Fridericia correction, which is given by
\[
QTc = \frac{QT}{\sqrt[3]{RR}}.
\]
The purpose of the statistical analysis in a TQT study is to formally assess if clinically relevant prolongation is present based on the QTc measurements \cite{patterson_bioequivalence_2006}.

TQT studies are predominantly conducted as cross-over trials in an attempt to lower sample size and eliminate between subject variation by paired comparisons of different treatments. In a cross-over trial, each subject is randomized to one of several treatment sequences that uniquely determines the treatment they receive at any given treatment period throughout the trial. Within each period, a baseline QTc measurement is obtained just prior to treatment, followed by a number of post treatment QTc measurements obtained at pre-defined time points following treatment (see Figure \ref{cross} for the two-period case). Each pair of consecutive periods will be separated by a washout period to minimize the risk of carry-over effects, i.e. any effects of treatment from the  previous period on the QTc measurements in the current period. From a practical point of view, the main challenge with cross-over designs is that the washout period needs to be tailored to the half-life of drug concentration to reflect proper washout of the drug. Specifically, if the half life is long, an even longer (typically 5 half lives) wash-out period is required in order to avoid carry-over effects. Ultimately, this may impose a very long study period for the subjects enrolled in the study. This is clearly not optimal and may also prove to be a challenge with regard to case retention. In such situations, a parallel arm design may be more feasible \cite{food_and_drug_administration_hhs_international_2005}. 
\begin{figure}[ht]
    \centering
    \includegraphics[width=\linewidth]{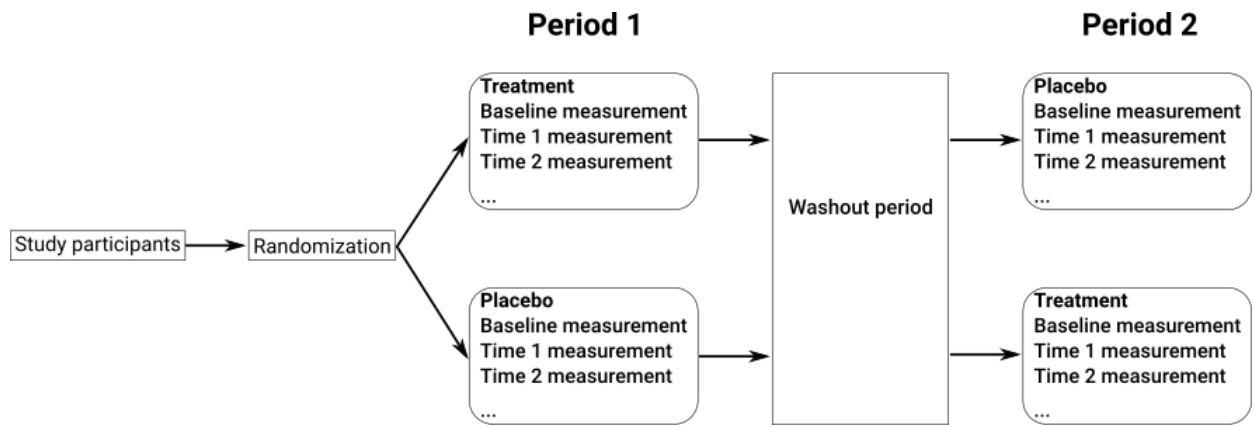}
    \caption{cross-over trial design.}
    \label{cross}
\end{figure}

The outline of the paper is as follows. In the next section we introduce the basic notation used in the paper and in that context define the fundamental assumption of no carry-over in a cross-over trial. We briefly introduce the concept of counterfactual outcomes and define the causal quantities of interest alongside the data assumptions needed to identify these quantities directly from the data. Section 3 is dedicated to deriving and assessing the theoretical performance of a number of causally motivated estimators. The section also contains the main result of this paper, which shows that a large class of working models still result in unbiased estimates of the causal target parameters even if they do not capture the true data generating mechanism. 
In Section 4, we analyse data from a real TQT study using a range of working models that in theory result in unbiased estimates of the causal target parameters according to the main result of this paper. Section 5 is dedicated to comparing the same working models in a simulation study cast around the data example in order to evaluate performance in a realistic scenario. We conclude the paper with a discussion in Section 6.

\section{Notation and assumptions}
In the following, we introduce the notation and causal assumptions needed in order to identify the target of estimation. Let $Y_{ipt}$ denote the QTc measurement for subject $i$ in period $p$ at time $t$, $i=1,\ldots, n, p=1,\ldots, P, t = 1,\ldots, T$. Denote baseline measurement for subject $i$ in period $p$ by $X_{ip}$, and treatment by $Z_{ip}$. TQT studies tend to have as many treatments as periods. Thus, we will denote treatments by $0,\ldots P-1$, where $Z_{ip}=0$ corresponds to subject $i$ receiving placebo in period $p$. Often we will suppress the $i$ since we assume the subjects are independent draws from the same distribution. Furthermore, let $Y_p = (Y_{p1},\ldots, Y_{pT})^T$ denote the vector of post-baseline measurements in period $p$.

In line with the informed choice of washout period in TQT studies, we assume the wash-out period has been sufficient to ensure no carry-over effects. Under this assumption, our data can be described by the Directed Acyclic Graph (DAG) in Figure \ref{DAG} in the two-period case. Note that the DAG has no arrows from baseline measurements to post-baseline measurements, since we do not expect a causal effect of the baseline measurements. Instead, we expect any association between baseline measurements and post-baseline measurements to arise from the latent variables, $W$, $W_1$, and $W_2$ from Figure \ref{DAG}. The latent variable $W$ reflects the dependence owing to measurements being from the same subject, whereas the latent variables, $W_1$ and $W_2$ reflect the dependence between measurements from the same period, i.e. temporary traits. Despite the lack of arrows between baseline and post-baseline measurements, it still makes sense to adjust for baseline measurements, for example in a regression model, because we do not observe the latent variables, in which case the baseline measurements act as proxies for the latent variables. 
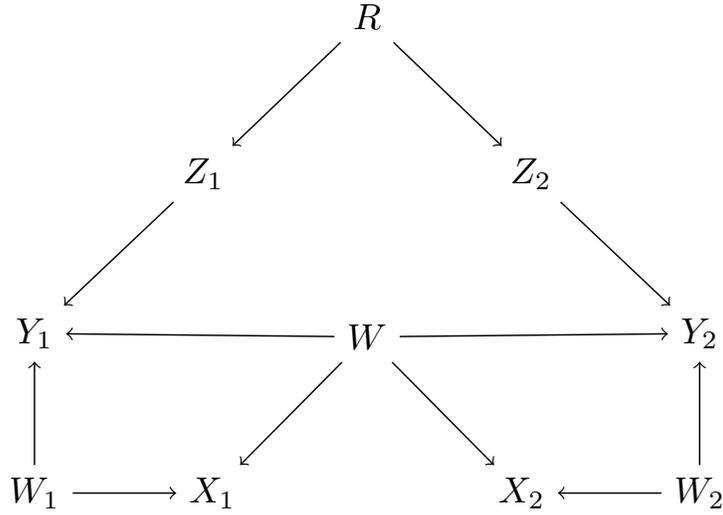
\begin{figure}[ht]
    \centering
    \resizebox{10cm}{!}{
    \begin{tikzpicture}
    \node (R) {$R$};
    \node[below left = of R] (B) {$Z_1$};
    \node[below right = of R] (E) {$Z_2$};
    \node[below left =  of B] (C) {$Y_1$};
    \node[below right = of E] (F) {$Y_2$};
    \node[below = of C] (W1) {$W_1$};
    \node[below = of F] (W2) {$W_2$};
    \node[below = 2.6cm of R] (W) {$W$};
    \node[right = of W1] (A) {$X_1$};
    \node[left = of W2] (D) {$X_2$};

    \draw[->] (B) -> (C);
    \draw[->] (E) -> (F);
    \draw[->] (W1) -> (C);
    \draw[->] (W2) -> (F);
    \draw[->] (W1) -> (A);
    \draw[->] (W2) -> (D);
    \draw[->] (R) -> (B);
    \draw[->] (R) -> (E);
    \draw[->] (W) -> (A);
    \draw[->] (W) -> (D);
    \draw[->] (W) -> (C);
    \draw[->] (W) -> (F);
    \end{tikzpicture}}
    \caption{$W$, $W_1$, $W_2$ = subject specific latent variables, R = Randomization.}
    \label{DAG}
\end{figure}

The DAG in Figure \ref{DAG} implies the following about the data distribution:
\begin{ass}[\bfseries No carry-over]
    Let $x = (x_1,\ldots, x_p)^T$, and likewise for $z$ and $y$, and let $f_x(x)$ be the density for variable $x$ and likewise for all other variables. 
    The distribution of our data satisfy the Markov factorization property with respect to the DAG in Figure \ref{DAG} \cite{peters_elements_2017}, i.e., the joint density of our data can be written as
    \begin{align*}
    f(z,w,x,y) =& f_z(z) f_w(w,w_1,\ldots, w_P) f_x(x|w,w_1,\ldots, w_P) f_y(y|z,w,w_1,\ldots, w_P) \\
    =& f_z(z) f_w(w,w_1,\ldots, w_P) \prod_{p=1}^P f_{x_p}(x_p|w,w_p) f_{y_p}(y_p|z_p,w,w_p).
    \end{align*}    
    \label{ass1}
\end{ass}
In accordance with Figure \ref{DAG}, Assumption \ref{ass1} states that the conditional distribution of $y$ in period $p$ only depends on treatment in period $p$ and the latent variables $w$ and $w_p$. 

In the following, we take a causal approach to TQT studies. By doing so, we provide a clearly stated research question that is completely disentangled from the modelling of the data. This exercise provides complete clarity on what assumptions about the data generating mechanism are necessary to answer the research question and sets them apart from purely technical assumptions made during the modelling stage of the estimation. 

Let $Y_p^{z_1,\ldots,z_P}$ denote the post-baseline QTc measurements we would have made in period $p$ if, possibly counter to fact, the subject had received treatments $z_1,\ldots, z_P$. Note that under Assumption \ref{ass1}, the QTc measurements in period $p$ only depend on the treatment in period $p$. I.e., $Y_p^{z_1,\ldots,z_p,\ldots,z_P} = Y_p^{q_1,\ldots,z_p,\ldots,q_P} = Y_p^{z_p}$ for all treatment regimes and periods. Thus, from here on, $Y_{p}^{z_p}$ will denote the potential QTc measurements in period $p$ if the subject, possibly counter to fact, had received treatment $z_p$ in period $p$.   

If we were particularly interested in treatment $z$ and had ample resources in terms of money, time, and subjects, we would have made a two-arm trial and used
\[
E(Y_t^z - Y_t^0), \quad t = 1,\ldots, T,
\]
as the causal contrasts of interest. Note that the first period in a cross-over trial corresponds to such a trial and as a consequence the targeted causal contrasts can be identified as:
\[
E(Y_{1t}^z - Y_{1t}^0), \quad t = 1,\ldots, T. 
\]
We can then easily estimate the contrasts based on data from the first period only, for example by
\[
    \frac{\sum_{i=1}^n I(Z_{i1} = z)Y_{i1t}}{\sum_{i=1}^n I(Z_{i1} = z)} - \frac{\sum_{i=1}^n I(Z_{i1} = 0)Y_{i1t}}{\sum_{i=1}^n I(Z_{i1} = 0)}, \quad t = 1,\ldots, T.
\]
Clearly this is not an efficient use of all the data collected in the cross-over trial. However, to enable full use of the data, stricter assumptions than Assumption \ref{ass1} are needed. Specifically, we need to make assumptions about the distributions of the post-baseline measurements. To this end, it would be natural to  assume the same treatment effect in all periods:
\begin{ass}[\bfseries Same treatment effect]
    \[
    E(Y_p^z - Y_p^0) = E(Y_q^z - Y_q^0) 
    \]
    for all $p$ and $q$. I.e. the treatment effect is the same in all periods. 
    \label{sameeffect}
\end{ass}
Assumption \ref{sameeffect} enables estimation of one overall treatment effect across all periods. A special case where Assumption \ref{sameeffect} holds is when the distribution of period specific data does not depend on period:
\begin{ass}[\bfseries Same distribution]
    \[
    (Y_p^0,\ldots, Y_p^{P-1},X_p,Z_p) \stackrel{\mathcal{D}}{=} (Y_q^0,\ldots, Y_q^{P-1},X_q,Z_q),
    \]    
    for all $p$ and $q$.
    \label{allsame}
\end{ass}
 
Assumption \ref{allsame} is rather restrictive compared to Assumption \ref{sameeffect} in that it a priori excludes any systematic effect due to period. This contradicts current modelling and design practice in cross-over trials, where potential period effects are modelled and subjects are randomized according to a latin square in order to balance any period effect \cite{senn_cross-over_2002}. 

As an alternative, one can assume that the conditional distribution of the post-baseline measurements, given baseline measurements and treatment, is the same in all periods:
\begin{ass}[\bfseries Same relationship]
\begin{equation*}
(Y_p | X_p=x, Z_p = z) \stackrel{\mathcal{D}}{=} (Y_q | X_q=x, Z_q = z)
\end{equation*}
for all $p$ and $q$. 
\label{samerelation}
\end{ass}
This facilitates a model fit based on data from all periods, which may in turn be used to infer the period specific causal contrast $E(Y_p^z - Y_p^0)$.

\section{Estimation}
The causal framework from the last section enables us to clearly define our target of estimation. In this section, we provide inference procedures tailored to assess the causal target parameter in the context of a TQT study. For brevity, we only consider the case, where we have assumed the same average causal treatment effect in all periods. I.e., we specifically develop estimators for $E(Y_1^z - Y_1^0)$ under Assumption \ref{sameeffect}. An outline of how to estimate period specific average causal effects without Assumption \ref{sameeffect} is provided in Section 6.  

Under Assumption \ref{sameeffect} the fact that subjects receive both the placebo and active treatment initially motivates the following simple non-parametric estimator: 
\begin{equation*}
    \hat{\mu}_{1t}(z) = \frac{1}{n} \sum_{i=1}^n \sum_{p=1}^P I(Z_{ip}=z)Y_{ipt} - I(Z_{ip}=0)Y_{ipt},
\end{equation*}
where the indicator function, $I(A)$, maps elements of $A$ to one and is zero otherwise. Note that this, and all other estimators throughout this paper,
has $t$ in the subscript to indicate that it is the estimator for the post-baseline measurement time point $t$. Naturally, an effect will be estimated for each $t=1,\ldots, T$. 
Note that $\hat{\mu}_{1t}(z)$ simply takes the outcomes in the periods, where the subjects receive the treatment of interest, and subtracts the outcomes in the placebo periods,
and averages across subjects. It is an unbiased estimator for the treatment effects of interest due to the randomization. However, it only uses post-baseline measurements, and may therefore lack precision.
Alternatively, one can pursue fitting a working regression model to the data and thereby bring baseline measurements into play. We specifically assume that the possible outcome predictions in this working model are given by   $h_{pt}(x,z,\beta)$, where $\beta$ is a vector of regression parameters.
Under Assumption \ref{allsame} it is possible to ignore periods and use the simpler working regression model $h_{t}(x,z,\beta)$.
Such a working model can be used to estimate the causal effects of interest simply by plugging into the G-computation formula \cite{robins_new_1986}:
\[
    \hat{\mu}_{2t}(z) = \frac{1}{n\cdot P}\sum_{i=1}^n\sum_{p=1}^P [h_{pt}(X_{ip}, z, \hat{\beta})-h_{pt}(X_{ip}, 0, \hat{\beta})].
\]
This estimator uses covariate information to gain efficiency. Moreover, in situations with missing endpoint data, the estimator is still unbiased under the missing at random assumption given that the regression model is correctly specified. In comparison, the endpoint data has to be missing completely at random for the simple estimator $\hat{\mu}_{1t}(z)$ to be unbiased.  

Up front, the above developments depend on the fact that the working regression model is specified so that it captures the true mean value structure. Since this by no means is warranted, it is important to mitigate the impact in terms of bias if the working model is misspecified. Such a mitigation  can be successfully achieved with the following semi-parametric estimator:
\[
    \hat{\mu}_{3t}(z) = \hat{\mu}_{1t}(z) - \frac{1}{n}\sum_{i=1}^n\sum_{p=1}^P \left[\left(I(Z_{ip} = z) - \frac{1}{P}\right)h_{pt}(X_{ip}, z, \hat{\beta}) - \left(I(Z_{ip} = 0) - \frac{1}{P}\right)h_{pt}(X_{ip}, 0, \hat{\beta})\right].
\]
The estimator is derived in Appendix C. It uses covariates to gain precision, but is unbiased due to the independence between $X_p$ and $Z_p$. This independence is ensured by the randomization and implies that the last term has a mean of zero. Thus, the estimator has the same mean as the non-parametric estimator, namely the true causal effect. The following theorem shows that $\hat{\mu}_{3t}(z) = \hat{\mu}_{2t}(z)$ for certain types of working models. 
\begin{theo}[\bfseries Unbiasedness of $\hat{\boldsymbol\mu}_{\mathbf{2t}}\mathbf{(z)}$]\label{mainresult}
Assume the data structure of this paper, and assume we use a working model with post baseline measurements as outcome, and with main effects of treatment specific to each post baseline time point. Let $Y_i$ be the vector of all post baseline measurements for subject $i$, and let $h(X_i,Z_i,\beta)$ be the vector of all predictions for subject $i$ from the working model. Assume the working model parameters, $\beta$, are estimated from the following weighted least squares estimating equation:
\begin{equation}\label{linear_restriction}
\sum_{i=1}^n D_i V^{-1} (Y_i - h_i(X_i,Z_i,\beta)) = 0,
\end{equation}
where $D_i$ is the design matrix for subject $i$, and $V$ is a weight matrix on the form
\begin{align}
    \begin{pmatrix}
    A & B & \cdots & B \\
    B & A & \cdots & B \\
    \vdots & \vdots & \ddots & \vdots \\
    B & B & \cdots & A
    \end{pmatrix},
    \label{correlationstructure}
\end{align}
where A and B are $T\times T$ matrices. If $A$ and $A-B$ are non-singular, then 
\[
\hat{\mu}_{2t}(z) = \hat{\mu}_{3t}(z).
\]
\end{theo}
\begin{proof}
The proof is provided in Appendix A.
\end{proof}
We know that $\hat{\mu}_{3t}(z)$ is unbiased by construction. Hence, Theorem \ref{mainresult} implies unbiased estimation when using $\hat{\mu}_{2t}(z)$ with a working model satisfying the conditions in Theorem \ref{mainresult}, regardless of the ability of the working model to capture the true data generating mechanism. 

Previously, the optimal choice of model for TQT studies has been debated extensively, based on the premise that such a model should provide a valid fit to the data. Theorem \ref{mainresult} implies the existence of a whole range of working models that may be used to provide unbiased estimates of causal treatment effects without such considerations. 

In this context, it is also worthwhile to note that many of the models that are traditionally applied in TQT studies satisfy the conditions in Theorem \ref{mainresult}. We summarize this in the following corollary. 
\begin{mycorollary}[\bfseries Gaussian linear mixed models]\label{cor1}
Assume the working model is a Gaussian linear mixed model with main effects of treatment specific to each post baseline time point, and a correlation structure satisfying (\ref{correlationstructure}). Then, $\hat{\mu}_{2t}(z) = \hat{\mu}_{3t}(z)$ if  model parameters are estimated by maximum likelihood or restricted maximum likelihood estimation. 
\end{mycorollary}
\begin{proof}
The estimating equation for Gaussian linear mixed models is given on page 10 of \cite{jiang_linear_2007} and can be rewritten to (\ref{linear_restriction}).
\end{proof}
When the working model is a Gaussian linear mixed model, the requirement (\ref{correlationstructure}) corresponds to modelling the dependence within periods by some matrix $A$, and the dependence between periods by a matrix $B$. For example, $B$ might be a matrix of constants corresponding to a random subject effect, and $A$ can be modelled more flexibly, for example with an unstructured covariance structure, or an AR(1) covariance structure. 

In the special case of a Gaussian linear mixed model, consider $A = \sigma^2 I$ and $B$ a matrix of zeros. In this case, the working model is a standard linear regression model that ignores any dependence between observations. Corollary \ref{cor1} then ensures that we are able to produce sound inference even with this simplistic model. We do, however, expect this working model to be less efficient than if we model the dependence structure in a Gaussian linear mixed model. 

A particular example of a much used working regression model is given by \cite{patterson_bioequivalence_2006}:
\begin{equation}
h_{pt}(x,z,\beta)=\beta_{pt} + \beta_{x}x + \beta_{zt}. \label{patjonesmodel}
\end{equation}
Clearly the conditions of Theorem 1 are satisfied for the systematic part of this model, as it includes a main effect of treatment specific to each post baseline time point. 
Additionally, the covariance structure proposed in \cite{patterson_bioequivalence_2006} is AR(1) within periods, and assuming constant covariance between observations on the same subject in different periods. Accordingly, the proposed covariance structure complies with (\ref{correlationstructure}). The estimates of the time specific effects of treatment $\beta_{zt}$ equal the estimates obtained if we were to plug model (\ref{patjonesmodel}) into $\hat{\mu}_{2t}(z)$. Thus, the estimates of $\beta_{zt}$ are unbiased for the treatment effects of interest under arbitrary model misspecification.

In order to enable inference fully, we further need to  characterize the large sample behaviour. This is well established if the targeted treatment effects appear as parameters in the model, and assuming that the model is correctly specified. However, in more complex models, the target treatment effect is not readily identified as a parameter specified in the model. Moreover, model based standard errors may not be appropriate, unless the model is correctly specified. The influence function of $\hat{\mu}_{2t}(z)$ is derived in Appendix B under the assumption that the $\beta$ parameters are estimated using an M-estimator. For models not covered by Theorem 1, $\hat{\mu}_{3t}(z)$ is still unbiased whereas $\hat{\mu}_{2t}(z)$ may be biased. Therefore, it would be preferable to use $\hat{\mu}_{3t}(z)$ in such cases. Accordingly, we also derive the influence function for $\hat{\mu}_{3t}(z)$ in Appendix C. 

One particular use of the above asymptotic results is when assessing QT prolongation in a TQT trial. In this context, the test for QTc prolongation for the drug of interest is carried out by use of an intersection-union test. I.e. the null-hypothesis is 
\[
H_0\colon\bigcup_{t=1}^T \mu_t(z) \geq \Delta,
\]
where 
\[
\mu_t(z) = E(Y_{pt}^z - Y_{pt}^0),
\]
and $\Delta$ is some reasonable amount of QT prolongation, such as 10 ms \cite{patterson_bioequivalence_2006}. Commonly speaking, the null-hypothesis dictates that there exists a time point where QT prolongation exceeds a prespecified clinically negligible threshold.  Tests are carried out for each $t$ based on the asymptotic behaviour of the standardized estimates of $\mu_t(z)$, and the null is rejected if all of these tests are rejected.

In addition, it is common practice to assess if prolongation can be detected for the positive control. The corresponding null-hypothesis is
\[
H_0\colon\bigcap_{t=1}^T \mu_t(z) \leq \Delta.
\]
The test of this hypothesis needs to be adjusted for multiple testing. This adjustment should be made in the most efficient way possible, which is possible because we can estimate the joint (asymptotic) distribution of our estimators \cite{pipper_versatile_2012,hothorn_simultaneous_2008}. We can derive the joint asymptotic distribution of our estimators from the influence functions as
\[
\sqrt{n}(\hat{\mu}_2(z) - \mu(z)) = 
\begin{pmatrix}
\hat{\mu}_{21} - \mu_{1}(z) \\
\vdots \\
\hat{\mu}_{2T} - \mu_{T}(z)
\end{pmatrix} = \frac{1}{\sqrt{n}} \sum_{i=1}^n \begin{pmatrix}
\varphi_1(d_i) \\
\vdots \\
\varphi_T(d_i)
\end{pmatrix} + o_p(1) = \frac{1}{\sqrt{n}} \sum_{i=1}^n \varphi(d_i) + o_p(1),
\]
where $\varphi_t(d_i), t=1,\ldots, T$ are the influence functions of the individual estimators. 
It follows that the asymptotic variance matrix of the joint distribution of our estimators equals
\begin{equation}\label{asympvar}
\frac{E(\varphi^T \varphi)}{n}.
\end{equation}
The estimation of targeted treatment effects as outlined above is based on well known regression models and as we have seen the targeted treatment effects may sometimes even be identified directly as regression parameters in those models. Theorem \ref{mainresult} then guarantees that the identified regression parameters are estimated without bias, irrespective of whether the regression model is correctly specified or not. This guarantee, however, does not extend to the model based variance matrix of the estimates. For this to be appropriate, one also needs to assume that the regression model is correct. The asymptotic variance matrix (\ref{asympvar}) on the other hand is generally applicable. When the targeted treatment effects are identified as regression parameters, it may even be obtained by standard software \cite{clubsandwich}.

The asymptotic theory developed is applicable, as is when the sample size is large. However, TQT trials and many cross-over trials have a rather small sample size, in which case the appropriateness of the asymptotic theory is questionable. For such cases, there is a substantial literature on how to improve asymptotic standard errors and confidence intervals \cite{bell2002bias,colin_cameron_practitioners_2015,mackinnon_heteroskedasticity-consistent_1985,pustejovsky_small-sample_2018}. One simple improvement of the asymptotic standard errors from \cite{colin_cameron_practitioners_2015} is to use the influence function in the same way as we would in the context of a linear normal model. I.e. to estimate the variance by $\frac{1}{n-1}\sum_{i=1}^n \hat{\varphi}_i^2$ instead of $\frac{1}{n}\sum_{i=1}^n \hat{\varphi}_i^2$, and use the 97.5\% quantile of the t-distribution with $n-1$ degrees of freedom instead of a standard normal distribution to construct confidence intervals. These modifications vanish as the sample size increases, and will consequently lead to correct asymptotic inference. These simple small sample modifications are used in the analyses and simulations throughout this paper.

\section{Data example}
To illustrate the developments in this paper, we reanalyse a standard TQT trial also analysed in chapter 9 of \cite{patterson_bioequivalence_2006}. The data set is freely available on the book website, and consists of 41 subjects, two of which we have excluded due to missingness. There are three single-dose treatments (C, D, E) and a placebo (F). Treatment E is included as a positive control, i.e. treatment E is known to mildly prolong the QT interval. The subjects' QT intervals are measured in triplicates at baseline, 0.5 hours, 1 hour, 1.5 hours, 2.5 hours, and 4 hours post treatment. The triplicates are averaged at each time point.

In accordance with the developments presented in  Section 3, we analyse the data with a range of models of differing complexity, that, in theory, facilitate unbiased estimates of the average causal effects under Assumption \ref{sameeffect}. We informally compare these models in terms of obtained estimates and standard errors. 

Our benchmark model corresponds to the recommendation made in \cite{lu_efficient_2014}. This paper advocates a regression model including average baseline measurements as a covariate. It is shown in \cite{lu_efficient_2014} that this approach is consistent with the joint baseline and post baseline measurement model advocated in \cite{kenward2010use,meng_use_2010}. It is further argued in \cite{lu_efficient_2014} that the resulting estimates of the treatment effects will be superior in terms of precision. 

Specifically, the working model proposed in \cite{lu_efficient_2014} is given by:
\[
h_{pt}(x,z,\beta)=\beta_{pt} + \beta_{xt} x + \beta_{\bar{x}t} \bar{x} + \beta_{zt}, 
\]
where $\bar{X}$ is the average baseline measurement. The effect of the baseline measurements and average baseline measurements are different at different time points. 

Moreover, we fit a simpler model:
\[
h_{pt}(x,z,\beta)=    \beta_{pt} + \beta_{xt} x + \beta_{zt}, 
\]
i.e. a model without average baseline measurements, but still with interaction between the effect of baseline and time point. 

Furthermore, we fit an even simpler model:
\[
h_{pt}(x,z,\beta)=    \beta_{t} + \beta_{x} x + \beta_{zt}, 
\]
i.e. without average baseline measurement, no interaction between time point and period, and the effect of the baseline measurement is the same at all time points. All the models above have the treatment effects as specific parameters. In order to illustrate the modelling flexibility facilitated by Theorem \ref{mainresult}, we fit a model with interaction between baseline and treatment:    
\[
h_{pt}(x,z,\beta)=        \beta_{t} + \beta_{xz} x + \beta_{zt}. 
\]
Note that with the complexity of this model, it is no longer possible to identify the average causal effect as a parameter in the regression model. Therefore, we can no longer rely on standard inference of regression models, but need to rely on the general inference procedures developed in Section 3. 

On top of specifying a working regression structure for the models, we also need to specify the working covariance structure. The working covariance structure has to be on the form (\ref{correlationstructure}), and in the following we will use a random subject effect corresponding to $B$ being a matrix of constants unless otherwise specified. We consider three different specifications for the $A$ matrix: 
\begin{enumerate}
    \item Unspecified: all the variances and covariances have to be estimated.
    \item AR(1): variance matrix is on the form:
    \[A=
        \sigma^2
        \begin{pmatrix}
         1 & \rho & \hdots & \rho^{T-1} \\
         \rho & 1 & \hdots & \rho^{T-2} \\
         \vdots & \vdots & \ddots & \vdots \\
         \rho^{T-1} & \rho^{T-2}  & \hdots & 1
        \end{pmatrix},
    \]
    where $\rho$ and $\sigma^2$ are parameters to be estimated.
    \item Independence: $A = \sigma^2 I$ , where $I$ is the identity matrix. In this case, $B$ will be a matrix of zeros corresponding to a standard linear regression model.
\end{enumerate}
Last, we also fit the non-parametric estimator $\hat{\mu}_{1t}(z)$.
The estimates of the effect of treatment E compared to placebo at post baseline time 4 from the models are displayed in Table \ref{example_means}. Note that the standard errors and confidence intervals are based on the small sample size adjustment discussed in the last section. The remaining estimates of treatment effects are presented in Appendix D. The different models result in similar treatment effects as expected. In practice, the main reason for choosing one model over another is therefore in order to have as much efficiency as possible, and not because we expect to actually know the true data-generating mechanism. 
\begin{table}[htb]
\ra{1.2}
\centering
\begin{tabular}{@{}llccc@{}} \toprule
Mean structure & Covariance structure & Estimate & Standard Error & 95\% CI \\ 
\midrule
$\beta_{pt} + \beta_{xt} x + \beta_{\bar{x}t} \bar{x} + \beta_{zt}$ & Unspecified & 8.32 & 1.53 & (5.22, 11.42) \\
 & AR(1) & 8.11 & 1.49 & (5.09, 11.14) \\
 & Independence & 8.19 & 1.54 & (5.07, 11.32) \\
 $\beta_{pt} + \beta_{xt} x + \beta_{zt}$ & Unspecified & 8.22 & 1.52 & (5.14, 11.31) \\
 & AR(1) & 8.09 & 1.48 & (5.11, 11.08) \\
 & Independence & 8.19 & 1.52 & (5.12, 11.26) \\
 $\beta_{t} + \beta_{x} x + \beta_{zt}$ & Unspecified & 8.49 & 1.43 & (5.59, 11.40) \\
 & AR(1) & 8.30 & 1.44 & (5.38, 11.22) \\
 & Independence & 8.43 & 1.44 & (5.51, 11.34) \\
 $\beta_{t} + \beta_{xz} x + \beta_{zt}$ & Unspecified & 8.43 & 1.40 & (5.60, 11.26) \\
 & AR(1) & 8.29 & 1.43 & (5.49, 11.10) \\
 & Independence & 8.42 & 1.42 & (5.55, 11.30) \\
 $\hat{\mu}_{1t}(z)$ & & 8.18 & 2.05 & (4.03, 12.33) \\
   \bottomrule
\end{tabular}
\caption{Estimates and standard errors for data example.}
\label{example_means}
\end{table}
From Table \ref{example_means} we note that estimates of the targeted treatment effect across all models seem comparable. The standard errors are substantially larger with the non-parametric approach, whereas for all other model based approaches, standard errors are comparable. We investigate these observations further in a simulation study mimicking the data example in the next section.

\section{Simulation study}
Simulation studies have already demonstrated that it is theoretically possible to gain precision by including the average baseline measurement as a covariate  \cite{lu_efficient_2014,meng_use_2010}. However, it is unclear how much the addition of the average baseline covariates matters for the precision in realistic setups. Therefore, we have based our simulation on the data set from the previous section. 

Specifically, we have simulated the data as follows: a joint normal distribution of the baseline measurements is fitted to the baseline measurements in the data set, and baseline measurements are simulated according to this fit. The average baseline model from the previous section is fitted to the data set, and the post baseline measurements are simulated from that model.

There are at least two advantages to this approach: first, the simulation must be considered realistic since it is based on parameters estimated on a data set from a real TQT trial. Second, we know that models ignoring the average baseline measurements are too simple to capture the true data-generating mechanism. Thereby, the simulations can show us how much precision we can expect to lose by not using average baseline measurements as covariates in real TQT studies.   

We compare all the models from the last section. The models are misspecified both in terms of mean structure and in terms of correlation structure, but are all unbiased by Corollary \ref{cor1}. 
We ran 10000 simulations in the statistical program R \cite{R}, and the code is available at \url{https://github.com/Jeepen/TQTpaper}. The results of the simulations are displayed in Table \ref{means}. Standard errors and confidence interval coverages are again based on an adjustment for small sample size. As expected, all estimators have negligible bias. 
\begin{table}[htb]
\centering
\ra{1.2}
\begin{tabular}{@{}llcccc@{}} \toprule
Mean structure & Covariance structure & Bias & SD & Avg. SE & Coverage \\ 
\midrule
$\beta_{pt} + \beta_{xt} x + \beta_{\bar{x}t} \bar{x} + \beta_{zt}$ & Unspecified & 0.02 & 1.48 & 1.43 & 0.947 \\
 & AR(1) & 0.02 & 1.48 & 1.43 & 0.947 \\
 & Independence & 0.02 & 1.48 & 1.48 & 0.954 \\
 $\beta_{pt} + \beta_{xt} x + \beta_{zt}$ & Unspecified & 0.01 & 1.49 & 1.45 & 0.945 \\
 & AR(1) & 0.01 & 1.49 & 1.45 & 0.945 \\
 & Independence & 0.01 & 1.52 & 1.52 & 0.952 \\
 $\beta_{t} + \beta_{x} x + \beta_{zt}$ & Unspecified & 0.02 & 1.48 & 1.46 & 0.951 \\
 & AR(1) & 0.02 & 1.48 & 1.46 & 0.951 \\
 & Independence & 0.02 & 1.51 & 1.51 & 0.954 \\
 $\beta_{t} + \beta_{xz} x + \beta_{zt}$ & Unspecified & 0.02 & 1.48 & 1.38 & 0.939 \\
 & AR(1) & 0.02 & 1.47 & 1.38 & 0.939 \\
 & Independence & 0.02 & 1.50 & 1.49 & 0.951 \\
 $\hat{\mu}_{1t}(z)$ & & 0.02 & 1.94 & 1.94 & 0.952 \\
   \bottomrule
\end{tabular}
\caption{Bias, standard deviation of estimates, and coverage of confidence intervals in simulations.}
\label{means}
\end{table}
The standard error estimates seem approximately correct compared to the sample standard deviation of the estimates, and the coverages of the confidence intervals are consequently close to 95 \%. The standard deviations in the table show that by far the majority of the gain in precision from using a model comes from the inclusion of the baseline measurements. Using the average baseline measurement as a covariate adds little precision, even over the linear regression model. However, the gain in precision from including the average baseline as a covariate and using a more flexible covariance structure than independence is for free, since Corollary \ref{cor1} implies unbiased estimation in any case. 

\section{Discussion}
We have introduced causal reasoning to the field of TQT studies. We have shown how typical choices of estimators can be given very clear causal interpretations in terms of the data, and not just in terms of specific models. Furthermore, we have shown that popular choices of working models in many circumstances yield unbiased estimates of causal parameters under arbitrary model misspecification. We have illustrated these results in a data example and a simulation study. 

However, the unbiasedness of the proposed estimators follows from the balancing induced by randomization. In practice, we may have missing data, which can invalidate the balancing initially ensured by the randomization. That said, the amount of missing data is typically negligible in TQT studies owing to the fact that they are most often conducted on healthy subjects who are paid to participate \cite{food_and_drug_administration_hhs_international_2005}. 
In cases with a non-negligible amount of missing data the working regression models can be fitted to the observed data with MLE assuming missingness at random (MAR), and distinct parameters for the missing data mechanism and for the outcome. As a consequence, it is straightforward to estimate the causal effects when using a linear mixed model, where the causal effects are identified as the main effects of treatment. However, this becomes challenging if we consider more complex models. In that case, a viable strategy could be to use imputation or weighting methods in order to go from model to estimates of the causal effects \cite{little_statistical_2020,tsiatis_semiparametric_2006}. 

We have focused on how to estimate causal effects under assumption (\ref{sameeffect}), that is,  assuming the effects are the same in all periods. To complement these developments, it is interesting to consider what we are estimating if that assumption does not hold. In general, the mean of $\hat{\mu}_{1t}(z)$ equals
\[
\frac{1}{P}\sum_{p=1}^P E(Y_{pt}^z - Y_{pt}^0),
\]
i.e. the average of the causal effects for each period. $\hat{\mu}_{3t}(z)$ has the same mean due to the randomization, and $\hat{\mu}_{2t}(z)$ will equal $\hat{\mu}_{3t}(z)$ when estimation is done according to Theorem 1. Thus, in general,  the above strategy will lead to unbiased estimation of the average of the period specific causal effects.   
Alternatively, one may fit a working model to all data under Assumption \ref{samerelation}, and subsequently apply G-computation for a single period:
\begin{equation}
\frac{1}{n}\sum_{i=1}^n h(X_{iPt}, z,\hat{\beta}) - h(X_{iPt}, 0,\hat{\beta}).
\label{one-period}
\end{equation}
The estimation of $h$ gains precision by using data from all periods, which in turn makes (\ref{one-period}) more precise. The estimator (\ref{one-period}) emulates a standard one-period trial, while it is unclear what we are emulating by ignoring the period specific treatment effects \cite{hernan_observational_2008}. This is a topic for further research.  

Two other theoretical issues also deserve more attention. First, we have not theoretically shown that $\hat{\mu}_{2t}(z)$ or for that matter $\hat{\mu}_{3t}(z)$ are in fact more efficient than the non-parametric estimator $\hat{\mu}_{1t}(z)$. We, however, suspect this to be the case in line with the results obtained for one-period trials in \cite{bartlett_covariate_2018,van_der_laan_targeted_2011}. Second, the impact  of a violation of restrictions on the working model dictated by Theorem 1 deserves further investigation. Possibly, a more flexible weight matrix than what is warranted by Theorem 1 may lead to further efficiency gains.

Finally, we would like to point out the difference between the estimation procedure proposed in this paper and the traditional approach of reporting treatment effects based on differences in least squares means  (see for example chapter 8 of \cite{patterson_bioequivalence_2006}). It is duly noted that differences in least squares means  are equivalent to $\hat{\mu}_{2t}(z)$ when the treatment effects are modelled as main effects in a linear mixed model. However, they are not equivalent to $\hat{\mu}_{2t}(z)$ when the model is more complex. In those cases, least squares means lack a proper causal interpretation in a meaningful population and on those grounds $\hat{\mu}_{2t}(z)$ or $\hat{\mu}_{3t}(z)$ should be preferred for assessing causal treatment effects.

\bibliographystyle{plain}
\bibliography{refs}

\section{Appendix A: Proof of Theorem 1}
We base the proof on rewriting the estimator $\hat{\mu}_{3t}(z)$ as
\begin{align}
    \hat{\mu}_{3t}(z) =& \frac{1}{n}\sum_{i=1}^n\sum_{p=1}^P I(Z_{ip} = z)Y_{ipt} - I(Z_{ip} = 0)Y_{ipt} \nonumber \\
    &- \frac{1}{n}\sum_{i=1}^n\sum_{p=1}^P \left[\left(I(Z_{ip} = z) - \frac{1}{P}\right)h_{pt}(X_{ip}, z, \hat{\beta}) - \left(I(Z_{ip} = 0) - \frac{1}{P}\right)h_{pt}(X_{ip}, 0, \hat{\beta})\right]\nonumber \\
    =&\frac{1}{n\cdot P}\sum_{i=1}^n\sum_{p=1}^P [h_{pt}(X_{ip}, z, \hat{\beta})-h_{pt}(X_{ip}, 0, \hat{\beta})] \nonumber \\
    &+ \frac{1}{n}\sum_{i=1}^n\sum_{p=1}^P I(Z_{ip}=z)(Y_{ipt} - h_{pt}(X_{ip}, z, \hat{\beta})) - I(Z_{ip}=0)(Y_{ipt} - h_{pt}(X_{ip}, 0, \hat{\beta})).\nonumber
\end{align}
The last rewriting shows that $\hat{\mu}_{3t}(z)$ is equal to $\hat{\mu}_{2t}(z)$ plus an adjustment term, i.e., $\hat{\mu}_{2t}(z)=\hat{\mu}_{3t}(z)$ if the last term is zero. 

We assume the $\beta$-parameters are estimated from the following weighted least squares estimating equation.
\begin{equation}
\sum_{i=1}^n D_i V^{-1} (Y_i - h(X_{i}, Z_{i}, \beta)) = 0, \label{estimatingequation}
\end{equation}
where $V$ is on the form (\ref{correlationstructure}).
The inverse of these matrices is also on the form
\begin{equation}
    \begin{pmatrix}
    A & B & \cdots & B \\
    B & A & \cdots & B \\
    \vdots & \vdots & \ddots & \vdots \\
    B & B & \cdots & A
    \end{pmatrix}^{-1} =
    \begin{pmatrix}
    C & D & \cdots & D \\
    D & C & \cdots & D \\
    \vdots & \vdots & \ddots & \vdots \\
    D & D & \cdots & C
    \end{pmatrix},
    \label{matrixshape}
\end{equation}
where C and D are matrices of the same dimensions as A and B, namely $T\times T$.
This can be realized by using the Woodbury identity (Theorem 18.2.8 in \cite{harville_matrix_2008}). The Woodbury identity states that assuming $X$ is nonsingular, then $X+Y$ is nonsingular if and only if $I + X^{-1}Y$ is nonsingular, and in that case:
\[
(X+Y)^{-1} = X^{-1} - X^{-1}Y(I + X^{-1}Y)^{-1}X^{-1}.
\]
For our purposes we will choose: 
\[
X = \begin{pmatrix}
A-B & 0 & \cdots & 0 \\
0 & A-B & \cdots & 0 \\
\vdots & \vdots & \ddots & \vdots \\
0 & 0 & \cdots & A-B
\end{pmatrix},
\]
and 
\[
Y = \begin{pmatrix}
B & B & \cdots & B \\
B & B & \cdots & B \\
\vdots & \vdots & \ddots & \vdots \\
B & B & \cdots & B
\end{pmatrix},
\]
such that $V = X + Y$. To realize that (\ref{matrixshape}) follows from the Woodbury identity, one needs to realize that $X^{-1}$ is block diagonal with $(A-B)^{-1}$ in the diagonal. Moreover, 
\[
X^{-1}Y = \begin{pmatrix}
(A-B)^{-1}B & (A-B)^{-1}B & \cdots & (A-B)^{-1}B \\
(A-B)^{-1}B & (A-B)^{-1}B & \cdots & (A-B)^{-1}B \\
\vdots & \vdots & \ddots & \vdots \\
(A-B)^{-1}B & (A-B)^{-1}B & \cdots & (A-B)^{-1}B
\end{pmatrix}.
\]
The inverse of $(I + X^{-1}Y)$ is then $EF$, where
\begin{align*}
&E = \begin{pmatrix}
(P-1)(A-B)^{-1}B + I & -(A-B)^{-1}B & \cdots & -(A-B)^{-1}B \\
-(A-B)^{-1}B & (P-1)(A-B)^{-1}B + I & \cdots & -(A-B)^{-1}B \\
\vdots & \vdots & \ddots & \vdots \\
-(A-B)^{-1}B & -(A-B)^{-1}B & \cdots & (P-1)(A-B)^{-1}B + I
\end{pmatrix}, \\
&F = \begin{pmatrix}
(P(A-B)^{-1}B+I)^{-1} & 0 & \cdots & 0 \\
0 & (P(A-B)^{-1}B+I)^{-1} & \cdots & 0 \\
\vdots & \vdots & \ddots & \vdots \\
0 & 0 & \cdots & (P(A-B)^{-1}B+I)^{-1}
\end{pmatrix}.
\end{align*}
This can be checked by multiplying the matrices and see that they give you the identity (it might be easier if you substitute $(A-B)^{-1}B$ with $K$). 
Then $X^{-1}Y(I + X^{-1}Y)^{-1}X^{-1}$ is a block matrix where all blocks are equal. 
Then the Woodbury identity gives us a block diagonal matrix minus a block matrix with all blocks equal, which is on the form (\ref{matrixshape}).   

For convenience, we will introduce
\[
\varepsilon_{ipt} = Y_{ipt} - h_{pt}(X_{ip}, Z_{ip}, \beta).
\]
There is an equation in (\ref{estimatingequation}) for each column in our design matrix, i.e. one equation for each $\beta$-parameter. However, as it turns out, we only need some equations to get $\hat{\mu}_{2t}(z) = \hat{\mu}_{3t}(z)$. We re-parametrize our model so that instead of having a $\beta$-parameter per combination of treatment and time point, we have main effects from time points, and main effects of non-placebo treatment per time point, i.e. placebo treatment becomes a reference level. First, we need the equations coming from having an effect of time point in our model. If we specifically look at the estimating equation coming from the effect of time point $t_0\in \{1,\ldots, T\}$, then it looks like this:  
\begin{align}
    0&=\sum_{i=1}^n\sum_{p=1}^P\sum_{t=1}^T\varepsilon_{ipt}\cdot (c_{tt_0} + (P-1)d_{tt_0}) \nonumber\\
    &=\sum_{t=1}^T(c_{tt_0} + (P-1)d_{tt_0})\sum_{i=1}^n\sum_{p=1}^P\varepsilon_{ipt}, \quad t_0=1,\ldots, T. \label{firstequation}
\end{align}
The inverse variance matrix leads to all residuals from all time points, periods, and subjects, being included in the estimating equation (\ref{firstequation}) for time point $t_0$. However, (\ref{firstequation}) is T linear equations, one for each $t_0$, with T unknowns, all equalling zero. Hence, we can conclude
\begin{equation}
\sum_{i=1}^n\sum_{p=1}^P\varepsilon_{ipt} = 0,\quad t=1,\ldots, T. \label{timecond}
\end{equation}
The second set of estimating equations we need comes from having main effects of treatment per time point in our model. The estimating equation corresponding to the effect of treatment $z$ at time point $t_0$ is:
\begin{align*}
    0 &= \sum_{i=1}^n\sum_{p=1}^P\sum_{t=1}^T\varepsilon_{ipt}\cdot (c_{tt_0}\cdot I(Z_{ip}=z) + d_{tt_0}\sum_{q\neq p} I(Z_{iq}=z)) \\ 
    &\stackrel{(*)}{=} \sum_{i=1}^n\sum_{p=1}^P\sum_{t=1}^T\varepsilon_{ipt}\cdot I(Z_{ip}=z)\cdot (c_{tt_0} - d_{tt_0}) + \varepsilon_{ipt}\cdot d_{tt_0}  \\
    &=\sum_{t=1}^T(c_{tt_0} - d_{tt_0})\sum_{i=1}^n\sum_{p=1}^P\varepsilon_{ipt}\cdot I(Z_{ip}=z) +\sum_{t=1}^T d_{tt_0} \sum_{i=1}^n\sum_{p=1}^P \varepsilon_{ipt} \\
    &\stackrel{(**)}{=}\sum_{t=1}^T(c_{tt_0} - d_{tt_0})\sum_{i=1}^n\sum_{p=1}^P\varepsilon_{ipt}\cdot  I(Z_{ip}=z),\quad t_0=1,\ldots, T.
\end{align*}
(*) comes from the fact that all subjects receive each treatment exactly once, so that $\sum_{q\neq p} I(Z_{iq}=z) = 1 - I(Z_{ip} = z)$. (**) comes from (\ref{timecond}). The rest is just interchanging the order of summation. Again, we have T equations with T unknowns and can conclude
\begin{equation}
    \sum_{i=1}^n\sum_{p=1}^P\varepsilon_{ipt}\cdot I(Z_{ip}=z) = 0, \quad t=1,\ldots, T, z = 1,\ldots, P-1, \label{treatcond}
\end{equation}
as wanted.
Note that the above argument only works for all other treatments than the placebo, since placebo is the reference treatment. However, we can rewrite (\ref{timecond}) to
\begin{align*}
0 =& \sum_{i=1}^n \sum_{p=1}^P \varepsilon_{ipt}\\
=& \sum_{i=1}^n \sum_{p=1}^P\sum_{z=0}^{P-1} I(Z_{ip}
= z)\varepsilon_{ipt} \\
=& \sum_{i=1}^n \sum_{p=1}^P I(Z_{ip}
= 0)\varepsilon_{ipt},
\end{align*}
where the last equality comes from (\ref{treatcond}). Thus, the extra term in $\hat{\mu}_{3t}(z)$ is zero, and $\hat{\mu}_{3t}(z) = \hat{\mu}_{2t}(z)$. Adding other covariates or terms to the regression will not change this fact as long as we have main effects of treatment specific to each post baseline time point. 

\section{Appendix B: Influence function for $\mathbf{\hat{\boldsymbol\mu}_{2t}(z)}$}
Let's say that the estimation in the first stage solves the following equation:
\begin{equation}
    \sum_{i=1}^n m(d_i, \beta) = 0,
    \label{mestimator}
\end{equation}
where $d_i$ is all the information we have about subject $i$. This would for example be the case if we used MLE in which case $m$ would be the score function. Denote the solution to (\ref{mestimator}) by $\hat{\beta}_n$, and the limit in probability of $\hat{\beta}_n$ by $\beta_0$.
Then the influence function of $\hat{\mu}_{2t}(z)$ is found using Theorem 6.1 from \cite{newey_chapter_1994}: 
\begin{equation}
    \varphi_t(d_i) = \frac{1}{P}\sum_{p=1}^P \left[h(X_{ip}, z, \beta_0)-h(X_{ip}, 0, \beta_0) - \mu_t(z)\right] + G_\beta \psi(d_i),
    \label{mu2influence}
\end{equation}
where
\begin{align*}
    &G_\beta = E\left(\nabla_\beta \frac{1}{P}\sum_{p=1}^P[h(X_{p}, z, \beta_0)-h(X_{ip}, 0, \beta_0)]\right), \\
    &\psi(d) = -(E(\nabla_\beta m(d,\beta_0))^{-1} m(d,\beta_0).
\end{align*}
The first term in (\ref{mu2influence}) represents the uncertainty coming from the covariate distribution, and is zero if the model simply is a main term linear mixed model, where $\hat{\mu}_{2t}(z)$ is a parameter in the model, and thereby independent of covariates. The second term represents the uncertainty arising from the estimation of the $\beta$-parameters in the first stage. The variance of the estimator is then 
\[
\frac{E(\varphi_t(d)^2)}{n}.
\]

\section{Appendix C: Derivation of $\mathbf{\hat{\boldsymbol\mu}_{3t}(z)}$}
The estimator is derived in the same way as the semi-parametric estimator for the pretest-posttest study in \cite{tsiatis_semiparametric_2006}, but with $\hat{\mu}_{1t}(z)$ as the non-parametric starting point. Note that this won't result in the efficient estimator in our setup, since we also have the assumption of the same treatment effect in all periods, which is not used in the derivation. 

The influence function of $\hat{\mu}_{1t}(z)$ is given by
\[
    \varphi_1(D) = \sum_{p=1}^P I(Z_p = z) Y_{pt} - I(Z_p = 0) Y_{pt} - \mu_t(z),
\]
where $\mu_t(z) = E(Y_{pt}^z - Y_{pt}^0)$ is the true causal effect, which we remind the reader is independent of period under Assumption \ref{sameeffect}, and $D$ consists of all the data from the subject. We derive the estimator $\hat{\mu}_{3t}(z)$ by the same calculations as \cite{tsiatis_semiparametric_2006}, namely
\begin{equation}
\varphi_3(D) = \varphi_1(D) - (E(\varphi_1(D) | \bar{X}, \bar{Z}) - E(\varphi_1(D) | \bar{X})), \label{influence}
\end{equation}
where $\bar{X} = \{X_1,\ldots, X_P\}$, and likewise for $\bar{Z}$. 
We calculate the expectations:
\begin{align*}
    E(\varphi_1(D) | \bar{X}, \bar{Z}) =& E\left(\sum_{p=1}^P I(Z_p = z) Y_{pt} - I(Z_p = 0) Y_{pt} - \mu_t(z) | \bar{X},\bar{Z}\right) \\
    =& \sum_{p=1}^P I(Z_p = z) E(Y_{pt} | \bar{X},\bar{Z}) - I(Z_p = 0) E(Y_{pt} |\bar{X},\bar{Z}) - \mu_t(z) \\
    =& \sum_{p=1}^P I(Z_p = z) E(Y_{pt} | X_p, Z_p = z) - I(Z_p = 0) E(Y_{pt} | X_p, Z_p = 0) - \mu_t(z),
\end{align*}
and
\begin{align*}
    E(\varphi_1(D) | \bar{X}) =& \sum_{p=1}^P \mathbb{P}(Z_p = z) E(Y_{pt} | X_p, Z_p = z) - \mathbb{P}(Z_p = 0) E(Y_{pt} | X_p, Z_p = 0) - \mu_t(z) \\
    =& \sum_{p=1}^P \frac{1}{P} E(Y_{pt} | X_p, Z_p = z) - \frac{1}{P} E(Y_{pt} | X_p, Z_p = 0) - \mu_t(z),
\end{align*}
The expectations $E(Y_{pt} | X_p, Z_p = z)$ and $E(Y_{pt} | X_p, Z_p = z)$ require a model, $h$.

When plugging the above into (\ref{influence}) we get an estimator of the influence function, $\varphi_3(D)$, and the estimator $\hat{\mu}_{3t}(z)$ is obtained by simple isolation from the equation
\[
\sqrt{n}(\hat{\mu}_{3t}(z) - \mu_t(z)) = \frac{1}{\sqrt{n}}\sum_{i=1}^n \varphi_3(D_i) + o_P(1),
\]
which defines influence functions. 

\subsection{Influence function for $\mathbf{\hat{\boldsymbol\mu}_{3t}(z)}$}
It might seem like the influence function was derived above. However, one could imagine that the estimation of the models in the first step would change the influence function, so we get a term like $G_\beta \psi(d)$ in the case of $\hat{\mu}_{2t}(z)$. 

We can simply take the estimator $\hat{\mu}_{3t}(z)$, subtract $\mu_t(z)$ and multiply by $\sqrt{n}$ in order to get
\begin{align}
\sqrt{n}(\hat{\mu}_{3t}(z) - \mu_t(z)) =& \frac{1}{\sqrt{n}} \sum_{i=1}^n   \sum_{p=1}^P I(Z_p = z) Y_{pt} - I(Z_p = 0) Y_{pt} - \mu_t(z) \nonumber \\ &- \left[\left(I(Z_{ip} = z) - \frac{1}{P}\right)h_{pt}(X_{ip}, z, \hat{\beta}_n) - \left(I(Z_{ip} = 0) 
- \frac{1}{P}\right)h_{pt}(X_{ip}, 0, \hat{\beta}_n)\right]. 
\label{biginfluence}
\end{align}
This looks like an equation that gives us the influence function, but note that $\hat{\beta}_n$ is estimated on all data, so the terms above are strictly speaking not independent. To proceed, we need to assume something about what happens to the model $h$ as $n$ gets bigger. Assume that there exists $\beta^*$ such that
\[
\sqrt{n}(\hat{\beta}_n - \beta^*)
\]
are bounded in probability and that $h$ as a function of $\beta$ is differentiable in a neighborhood of $\beta^*$. Then it is possible to Taylor expand $h$ in (\ref{biginfluence}) to get
\begin{align*}
    \sqrt{n}(\hat{\mu}_{3t}(z) - \mu_t(z)) =& \frac{1}{\sqrt{n}} \sum_{i=1}^n   \sum_{p=1}^P I(Z_p = z) Y_{pt} - I(Z_p = 0) Y_{pt} - \mu_t(z) \nonumber \\ &- \left[\left(I(Z_{ip} = z) - \frac{1}{P}\right)\left(h_{pt}(X_{ip}, z, \beta^*) + 0.5\frac{\partial h_{pt}(X_{ip}, z, \beta^*)}{\partial \beta}(\hat{\beta}_n - \beta^*) \right) \right. \\
    &\left.- \left(I(Z_{ip} = 0) 
- \frac{1}{P}\right)\left(h_{pt}(X_{ip}, 0, \beta^*) + 0.5\frac{\partial h_{pt}(X_{ip}, 0, \beta^*)}{\partial \beta}(\hat{\beta}_n - \beta^*) \right)  \right] + o_p(1) \\
=&\frac{1}{\sqrt{n}} \sum_{i=1}^n   \sum_{p=1}^P I(Z_p = z) Y_{pt} - I(Z_p = 0) Y_{pt} - \mu_t(z) \nonumber \\ &- \left[\left(I(Z_{ip} = z) - \frac{1}{P}\right)h_{pt}(X_{ip}, z, \beta^*) - \left(I(Z_{ip} = 0) 
- \frac{1}{P}\right)h_{pt}(X_{ip}, 0, \beta^*)\right] + o_p(1), 
\end{align*}
which gives us the influence function for $\hat{\mu}_{3t}(z)$. It looks a lot like what we would expect from (\ref{biginfluence}) with the detail that we have $h_{pt}(X_{ip}, z, \beta^*)$ and $h_{pt}(X_{ip}, 0, \beta^*)$ instead of $h_{pt}(X_{ip}, z, \hat{\beta}_n)$ and $h_{pt}(X_{ip}, 0, \hat{\beta}_n)$.
The influence function for $\hat{\mu}_{3t}(z)$ has the added advantage that it is easier to implement a variance estimator using that influence function than the influence function for $\hat{\mu}_{2t}(z)$. Therefore, it is preferable to use the influence function for $\hat{\mu}_{3t}(z)$ to estimate the variance when the two estimators coincide.

\section{Appendix D: Extra tables}
In the following tables, columns correspond to the models also used for the data example in the same order. Names are removed in the interest of space. 
\begin{table}[ht]
\centering
\begin{tabular}{ccrrrrrrrrrrrrr}
  \hline
Treatment & Time & 1 & 2 & 3 & 4 & 5 & 6 & 7 & 8 & 9 & 10 & 11 & 12 & 13 \\ 
  \hline
C & 0.5 & 3.76 & 3.71 & 3.75 & 4.04 & 3.97 & 4.19 & 4.01 & 3.96 & 4.16 & 3.94 & 3.89 & 4.09 & 2.66 \\ 
& 1.0 & 7.95 & 7.69 & 7.83 & 7.93 & 7.60 & 7.93 & 7.96 & 7.58 & 7.94 & 7.96 & 7.58 & 7.94 & 6.44 \\ 
& 1.5 & 5.62 & 5.59 & 5.59 & 5.68 & 5.71 & 5.90 & 5.67 & 5.69 & 5.88 & 5.67 & 5.69 & 5.88 & 4.38 \\ 
& 2.5 & 3.99 & 3.47 & 3.48 & 4.08 & 3.58 & 3.77 & 4.17 & 3.65 & 3.83 & 4.16 & 3.65 & 3.83 & 2.33 \\ 
& 4.0 & 4.32 & 4.70 & 4.72 & 4.53 & 4.86 & 5.05 & 4.57 & 4.86 & 5.03 & 4.55 & 4.86 & 5.03 & 3.53 \\ 
D & 0.5 & 5.11 & 6.04 & 6.03 & 5.05 & 6.05 & 6.05 & 5.31 & 6.22 & 6.23 & 4.93 & 5.89 & 5.83 & 5.95 \\ 
& 1.0 & 9.80 & 9.96 & 10.09 & 9.90 & 9.94 & 10.09 & 9.95 & 9.95 & 10.11 & 9.96 & 9.94 & 10.11 & 9.83 \\ 
& 1.5 & 7.39 & 7.20 & 7.20 & 7.28 & 7.20 & 7.21 & 7.50 & 7.36 & 7.38 & 7.50 & 7.35 & 7.37 & 7.09 \\ 
& 2.5 & 6.05 & 5.74 & 5.76 & 6.06 & 5.74 & 5.77 & 6.67 & 6.38 & 6.45 & 6.63 & 6.37 & 6.44 & 6.16 \\ 
& 4.0 & 5.79 & 5.53 & 5.66 & 5.69 & 5.51 & 5.68 & 5.98 & 5.76 & 5.96 & 5.89 & 5.75 & 5.95 & 5.68 \\ 
E & 0.5 & 0.88 & 1.41 & 1.41 & 0.78 & 1.41 & 1.41 & 1.01 & 1.61 & 1.62 & 0.75 & 1.37 & 1.42 & 1.38 \\ 
& 1.0 & 7.16 & 7.19 & 7.34 & 7.24 & 7.15 & 7.34 & 7.19 & 7.09 & 7.31 & 7.19 & 7.08 & 7.30 & 7.06 \\ 
& 1.5 & 6.03 & 6.03 & 6.08 & 5.92 & 6.02 & 6.08 & 6.13 & 6.15 & 6.24 & 6.13 & 6.15 & 6.24 & 5.99 \\ 
& 2.5 & 6.87 & 6.50 & 6.56 & 6.87 & 6.49 & 6.56 & 7.43 & 7.09 & 7.20 & 7.41 & 7.09 & 7.19 & 6.95 \\ 
& 4.0 & 8.32 & 8.11 & 8.19 & 8.22 & 8.09 & 8.19 & 8.49 & 8.30 & 8.43 & 8.43 & 8.29 & 8.42 & 8.18 \\ 
   \hline
\end{tabular}
\caption{Estimates in data example.}
\end{table}

\begin{table}[ht]
\centering
\begin{tabular}{ccrrrrrrrrrrrrr}
  \hline
Treatment & Time & 1 & 2 & 3 & 4 & 5 & 6 & 7 & 8 & 9 & 10 & 11 & 12 & 13 \\ 
  \hline
C & 0.5 & 1.21 & 1.26 & 1.30 & 1.22 & 1.28 & 1.35 & 1.23 & 1.28 & 1.31 & 1.24 & 1.28 & 1.29 & 1.81 \\ 
& 1.0 & 1.08 & 1.21 & 1.25 & 1.08 & 1.22 & 1.24 & 1.11 & 1.26 & 1.25 & 1.22 & 1.26 & 1.24 & 2.02 \\ 
& 1.5 & 1.13 & 1.08 & 1.10 & 1.12 & 1.07 & 1.08 & 1.15 & 1.10 & 1.09 & 1.01 & 1.07 & 1.07 & 1.76 \\ 
& 2.5 & 1.08 & 1.12 & 1.16 & 1.07 & 1.09 & 1.13 & 1.10 & 1.13 & 1.17 & 1.13 & 1.14 & 1.15 & 1.83 \\ 
& 4.0 & 1.35 & 1.36 & 1.40 & 1.34 & 1.35 & 1.39 & 1.39 & 1.37 & 1.40 & 1.34 & 1.39 & 1.38 & 2.03 \\ 
D & 0.5 & 1.01 & 1.00 & 1.05 & 1.09 & 1.12 & 1.27 & 1.07 & 1.15 & 1.26 & 1.15 & 1.15 & 1.25 & 1.48 \\ 
& 1.0 &  1.21 & 1.23 & 1.25 & 1.23 & 1.23 & 1.26 & 1.14 & 1.12 & 1.14 & 1.06 & 1.10 & 1.13 & 1.90 \\ 
& 1.5 & 1.13 & 1.14 & 1.18 & 1.14 & 1.15 & 1.24 & 1.10 & 1.11 & 1.17 & 1.10 & 1.12 & 1.16 & 1.88 \\ 
& 2.5 & 1.34 & 1.37 & 1.42 & 1.36 & 1.40 & 1.50 & 1.33 & 1.37 & 1.47 & 1.37 & 1.40 & 1.45 & 1.89 \\ 
& 4.0 & 1.31 & 1.40 & 1.45 & 1.30 & 1.40 & 1.46 & 1.27 & 1.37 & 1.40  & 1.34 & 1.37 & 1.38 & 2.10 \\ 
E & 0.5 & 1.41 & 1.38 & 1.42 & 1.39 & 1.38 & 1.42 & 1.30 & 1.29 & 1.31 & 1.27 & 1.28 & 1.29 & 1.79 \\ 
& 1.0 &  1.20 & 1.22 & 1.24 & 1.21 & 1.22 & 1.24 & 1.14 & 1.14 & 1.14 & 1.07 & 1.11 & 1.12 & 1.68 \\ 
& 1.5 & 1.31 & 1.35 & 1.40 & 1.30 & 1.33 & 1.39 & 1.28 & 1.31 & 1.33 & 1.30 & 1.33 & 1.31 & 1.97 \\ 
& 2.5 &  1.21 & 1.26 & 1.31 & 1.18 & 1.24 & 1.28 & 1.32 & 1.40 & 1.41 & 1.38 & 1.40 & 1.39 & 2.04 \\ 
& 4.0 &  1.53 & 1.49 & 1.54 & 1.52 & 1.48 & 1.52 & 1.43 & 1.44 & 1.44 & 1.40 & 1.43 & 1.42 & 2.05 \\ 
   \hline
\end{tabular}
\caption{Standard error estimates in data example.}
\end{table}

\begin{table}[ht]
\centering
\begin{tabular}{ccrrrrrrrrrrrrr}
  \hline
Treatment & Time & 1 & 2 & 3 & 4 & 5 & 6 & 7 & 8 & 9 & 10 & 11 & 12 & 13 \\ 
  \hline
C & 0.5 & 0.00 & 0.00 & 0.00 & 0.00 & 0.00 & 0.00 & 0.00 & 0.00 & 0.00 & -0.02 & -0.01 & -0.02 & 0.01 \\ 
  & 1.0 & 0.00 & 0.00 & 0.00 & 0.00 & 0.00 & 0.00 & 0.00 & 0.00 & 0.00 & 0.00 & 0.00 & 0.00 & 0.01 \\ 
& 1.5 &  -0.01 & -0.01 & -0.01 & -0.01 & -0.01 & -0.01 & -0.01 & -0.01 & -0.01 & -0.01 & -0.01 & -0.01 & -0.00 \\ 
& 2.5 &  0.01 & 0.01 & 0.01 & 0.01 & 0.01 & 0.01 & 0.00 & 0.00 & 0.00 & 0.00 & 0.00 & 0.00 & 0.01 \\ 
& 4.0 &  0.01 & 0.01 & 0.01 & 0.01 & 0.01 & 0.01 & 0.01 & 0.01 & 0.01 & 0.01 & 0.01 & 0.01 & 0.02 \\ 
D & 0.5 &  0.01 & 0.01 & 0.01 & 0.02 & 0.02 & 0.02 & 0.02 & 0.02 & 0.02 & -0.01 & -0.01 & -0.02 & 0.02 \\ 
& 1.0 &  0.01 & 0.01 & 0.01 & 0.01 & 0.01 & 0.01 & 0.01 & 0.01 & 0.01 & 0.01 & 0.01 & 0.01 & 0.02 \\ 
& 1.5 &  -0.01 & -0.01 & -0.01 & -0.01 & -0.01 & -0.02 & -0.01 & -0.01 & -0.01 & -0.01 & -0.01 & -0.01 & -0.01 \\ 
& 2.5 &  -0.01 & -0.01 & -0.01 & -0.01 & -0.01 & -0.01 & -0.02 & -0.02 & -0.02 & -0.02 & -0.02 & -0.02 & -0.01 \\ 
& 4.0 &  -0.00 & -0.00 & -0.00 & -0.00 & -0.00 & -0.00 & -0.00 & -0.00 & -0.00 & -0.01 & -0.01 & -0.00 & -0.00 \\ 
E & 0.5 &  0.02 & 0.02 & 0.02 & 0.02 & 0.02 & 0.02 & 0.02 & 0.02 & 0.02 & 0.02 & 0.02 & 0.01 & 0.03 \\ 
& 1.0 &  0.00 & 0.00 & 0.00 & 0.01 & 0.01 & 0.00 & 0.01 & 0.01 & 0.01 & 0.01 & 0.01 & 0.01 & 0.01 \\ 
& 1.5 &  0.00 & 0.00 & 0.00 & 0.00 & 0.00 & 0.00 & 0.00 & 0.00 & -0.00 & 0.00 & 0.00 & -0.00 & 0.01 \\ 
& 2.5 &  -0.02 & -0.02 & -0.02 & -0.02 & -0.02 & -0.02 & -0.02 & -0.02 & -0.02 & -0.02 & -0.02 & -0.02 & -0.01 \\ 
& 4.0 &  0.02 & 0.02 & 0.02 & 0.01 & 0.01 & 0.01 & 0.02 & 0.02 & 0.02 & 0.02 & 0.02 & 0.02 & 0.02 \\ 
   \hline
\end{tabular}
\caption{Bias in simulations.}
\end{table}

\begin{table}[ht]
\centering
\begin{tabular}{ccrrrrrrrrrrrrr}
  \hline
Treatment & Time & 1 & 2 & 3 & 4 & 5 & 6 & 7 & 8 & 9 & 10 & 11 & 12 & 13 \\ 
  \hline
C & 0.5 & 1.45 & 1.45 & 1.45 & 1.47 & 1.47 & 1.51 & 1.45 & 1.45 & 1.48 & 1.44 & 1.44 & 1.47 & 1.92 \\ 
& 1.0 & 1.22 & 1.22 & 1.22 & 1.22 & 1.22 & 1.22 & 1.21 & 1.21 & 1.22 & 1.21 & 1.21 & 1.22 & 1.97 \\ 
& 1.5 &  1.25 & 1.25 & 1.25 & 1.26 & 1.26 & 1.29 & 1.25 & 1.24 & 1.27 & 1.25 & 1.24 & 1.27 & 1.89 \\ 
& 2.5 &  1.42 & 1.42 & 1.42 & 1.42 & 1.42 & 1.44 & 1.43 & 1.43 & 1.46 & 1.43 & 1.43 & 1.46 & 1.95 \\ 
& 4.0 &  1.47 & 1.47 & 1.47 & 1.48 & 1.48 & 1.51 & 1.47 & 1.47 & 1.50 & 1.47 & 1.47 & 1.49 & 1.95 \\ 
D & 0.5 &  1.45 & 1.45 & 1.45 & 1.48 & 1.48 & 1.51 & 1.45 & 1.45 & 1.48 & 1.45 & 1.45 & 1.48 & 1.92 \\ 
& 1.0 &  1.23 & 1.23 & 1.23 & 1.22 & 1.22 & 1.23 & 1.21 & 1.21 & 1.22 & 1.21 & 1.21 & 1.22 & 1.96 \\ 
& 1.5 &  1.25 & 1.25 & 1.25 & 1.26 & 1.26 & 1.29 & 1.25 & 1.25 & 1.27 & 1.25 & 1.25 & 1.27 & 1.90 \\ 
& 2.5 &  1.41 & 1.41 & 1.41 & 1.41 & 1.41 & 1.44 & 1.44 & 1.44 & 1.47 & 1.44 & 1.44 & 1.47 & 1.94 \\ 
& 4.0 &  1.47 & 1.47 & 1.47 & 1.48 & 1.48 & 1.51 & 1.47 & 1.47 & 1.50 & 1.47 & 1.47 & 1.50 & 1.94 \\ 
E & 0.5 &  1.44 & 1.44 & 1.44 & 1.47 & 1.47 & 1.51 & 1.45 & 1.45 & 1.48 & 1.44 & 1.44 & 1.47 & 1.92 \\ 
 & 1.0 & 1.24 & 1.24 & 1.24 & 1.23 & 1.24 & 1.24 & 1.23 & 1.23 & 1.23 & 1.23 & 1.23 & 1.23 & 1.96 \\ 
 & 1.5 & 1.25 & 1.25 & 1.25 & 1.27 & 1.26 & 1.29 & 1.25 & 1.25 & 1.27 & 1.25 & 1.25 & 1.27 & 1.89 \\ 
 & 2.5 & 1.41 & 1.41 & 1.41 & 1.42 & 1.42 & 1.44 & 1.44 & 1.43 & 1.46 & 1.43 & 1.43 & 1.46 & 1.93 \\ 
 & 4.0 & 1.48 & 1.48 & 1.48 & 1.49 & 1.49 & 1.52 & 1.48 & 1.48 & 1.51 & 1.48 & 1.47 & 1.50 & 1.94 \\ 
   \hline
\end{tabular}
\caption{Standard deviation of estimates in simulations.}
\end{table}

\begin{table}[ht]
\centering
\begin{tabular}{ccrrrrrrrrrrrrr}
  \hline
Treatment & Time & 1 & 2 & 3 & 4 & 5 & 6 & 7 & 8 & 9 & 10 & 11 & 12 & 13 \\ 
  \hline
C & 0.5 &1.42 & 1.42 & 1.47 & 1.45 & 1.44 & 1.52 & 1.44 & 1.44 & 1.49 & 1.37 & 1.37 & 1.47 & 1.95 \\ 
& 1.0 &  1.19 & 1.19 & 1.23 & 1.20 & 1.20 & 1.23 & 1.20 & 1.21 & 1.23 & 1.11 & 1.16 & 1.21 & 1.96 \\ 
& 1.5 &  1.21 & 1.21 & 1.25 & 1.23 & 1.22 & 1.29 & 1.23 & 1.23 & 1.27 & 1.14 & 1.18 & 1.25 & 1.90 \\ 
& 2.5 &  1.37 & 1.37 & 1.41 & 1.38 & 1.37 & 1.44 & 1.41 & 1.41 & 1.46 & 1.33 & 1.35 & 1.44 & 1.94 \\ 
& 4.0 &  1.43 & 1.43 & 1.48 & 1.45 & 1.45 & 1.52 & 1.46 & 1.46 & 1.51 & 1.38 & 1.38 & 1.49 & 1.94 \\ 
D & 0.5 & 1.42 & 1.42 & 1.47 & 1.45 & 1.45 & 1.52 & 1.44 & 1.44 & 1.49 & 1.37 & 1.37 & 1.47 & 1.95 \\ 
 & 1.0 & 1.19 & 1.19 & 1.23 & 1.19 & 1.19 & 1.23 & 1.20 & 1.20 & 1.23 & 1.10 & 1.15 & 1.21 & 1.96 \\ 
& 1.5 &  1.22 & 1.22 & 1.25 & 1.23 & 1.23 & 1.29 & 1.24 & 1.24 & 1.28 & 1.15 & 1.19 & 1.26 & 1.90 \\ 
& 2.5 &  1.37 & 1.37 & 1.41 & 1.38 & 1.38 & 1.44 & 1.43 & 1.43 & 1.48 & 1.34 & 1.36 & 1.45 & 1.95 \\ 
 & 4.0 & 1.43 & 1.43 & 1.48 & 1.45 & 1.45 & 1.52 & 1.46 & 1.46 & 1.51 & 1.38 & 1.38 & 1.49 & 1.94 \\ 
E & 0.5 & 1.42 & 1.42 & 1.47 & 1.45 & 1.45 & 1.53 & 1.45 & 1.45 & 1.50 & 1.37 & 1.37 & 1.47 & 1.95 \\ 
 & 1.0 & 1.19 & 1.19 & 1.23 & 1.19 & 1.19 & 1.23 & 1.20 & 1.20 & 1.23 & 1.10 & 1.15 & 1.21 & 1.96 \\ 
 & 1.5 & 1.21 & 1.21 & 1.25 & 1.23 & 1.23 & 1.29 & 1.23 & 1.23 & 1.27 & 1.14 & 1.18 & 1.25 & 1.89 \\ 
 & 2.5 & 1.36 & 1.36 & 1.41 & 1.37 & 1.37 & 1.44 & 1.41 & 1.41 & 1.46 & 1.33 & 1.35 & 1.44 & 1.93 \\ 
 & 4.0 & 1.43 & 1.43 & 1.48 & 1.45 & 1.45 & 1.52 & 1.46 & 1.46 & 1.51 & 1.38 & 1.38 & 1.49 & 1.94 \\ 
   \hline
\end{tabular}
\caption{Average standard error estimates in simulations.}
\end{table}

\begin{table}[ht]
\centering
\begin{tabular}{ccrrrrrrrrrrrrr}
  \hline
Treat & Time & 1 & 2 & 3 & 4 & 5 & 6 & 7 & 8 & 9 & 10 & 11 & 12 & 13 \\ 
  \hline
C & 0.5 & 0.946 & 0.946 & 0.954 & 0.947 & 0.946 & 0.953 & 0.951 & 0.951 & 0.953 & 0.938 & 0.938 & 0.950 & 0.952 \\ 
& 1.0 &  0.948 & 0.948 & 0.954 & 0.946 & 0.947 & 0.953 & 0.951 & 0.951 & 0.953 & 0.937 & 0.938 & 0.951 & 0.955 \\ 
& 1.5 &  0.946 & 0.946 & 0.953 & 0.946 & 0.947 & 0.953 & 0.951 & 0.951 & 0.953 & 0.940 & 0.939 & 0.951 & 0.957 \\ 
& 2.5 &  0.948 & 0.948 & 0.954 & 0.949 & 0.948 & 0.955 & 0.951 & 0.951 & 0.953 & 0.927 & 0.940 & 0.950 & 0.949 \\ 
& 4.0 &  0.946 & 0.946 & 0.952 & 0.947 & 0.947 & 0.952 & 0.952 & 0.952 & 0.953 & 0.924 & 0.939 & 0.950 & 0.952 \\ 
D & 0.5 & 0.940 & 0.940 & 0.948 & 0.941 & 0.942 & 0.948 & 0.947 & 0.947 & 0.950 & 0.923 & 0.936 & 0.948 & 0.950 \\ 
& 1.0 &  0.945 & 0.945 & 0.952 & 0.945 & 0.945 & 0.953 & 0.951 & 0.951 & 0.954 & 0.930 & 0.940 & 0.950 & 0.951 \\ 
 & 1.5 & 0.944 & 0.944 & 0.951 & 0.946 & 0.945 & 0.952 & 0.949 & 0.949 & 0.954 & 0.930 & 0.940 & 0.950 & 0.950 \\ 
 & 2.5 & 0.947 & 0.947 & 0.952 & 0.946 & 0.946 & 0.953 & 0.949 & 0.949 & 0.952 & 0.928 & 0.938 & 0.948 & 0.951 \\ 
 & 4.0 & 0.945 & 0.945 & 0.952 & 0.945 & 0.945 & 0.952 & 0.948 & 0.948 & 0.952 & 0.932 & 0.935 & 0.947 & 0.949 \\ 
 E & 0.5 & 0.944 & 0.944 & 0.950 & 0.944 & 0.943 & 0.950 & 0.948 & 0.948 & 0.950 & 0.932 & 0.937 & 0.947 & 0.954 \\ 
 & 1.0 & 0.941 & 0.941 & 0.948 & 0.943 & 0.942 & 0.949 & 0.946 & 0.947 & 0.951 & 0.931 & 0.935 & 0.948 & 0.952 \\ 
 & 1.5 & 0.946 & 0.946 & 0.954 & 0.946 & 0.947 & 0.955 & 0.951 & 0.950 & 0.956 & 0.937 & 0.938 & 0.953 & 0.949 \\ 
 & 2.5 & 0.944 & 0.944 & 0.951 & 0.947 & 0.947 & 0.955 & 0.953 & 0.953 & 0.956 & 0.935 & 0.937 & 0.952 & 0.954 \\ 
 & 4.0 & 0.947 & 0.947 & 0.954 & 0.945 & 0.945 & 0.952 & 0.951 & 0.951 & 0.954 & 0.939 & 0.939 & 0.951 & 0.952 \\ 
   \hline
\end{tabular}
\caption{Coverage frequencies of confidence intervals in simulations.}
\end{table}
\end{document}